\documentclass[conference]{IEEEtran}
\IEEEoverridecommandlockouts
\usepackage[colorlinks,
linkcolor=red,
anchorcolor=green,
citecolor=blue
]{hyperref} 
\usepackage{amsmath,amssymb}
\usepackage{latexsym}
\usepackage{CJK}
\usepackage{cite}

\usepackage{color, soul}

\usepackage{graphicx}
\usepackage{epstopdf}
\usepackage{epsfig}
\usepackage{subfigure} 

\usepackage{verbatim}  

\usepackage{mathrsfs}	
\usepackage{algorithm} 
\usepackage{algorithmic} 
\usepackage{booktabs}
\usepackage{textcomp}
\usepackage{multirow}
\usepackage{lettrine}

\usepackage{mathrsfs}
\usepackage{amsmath}
\usepackage{amssymb}

\usepackage{algorithm}
\usepackage{algorithmic}
\usepackage{extarrows}

\usepackage{amsthm} 

\newtheorem{cor}{Corollary}
\newtheorem{prop}{Proposition}
\newtheorem{defn}{Definition}

{ \bgroup
  \addtolength\abovedisplayshortskip{#1}
  \addtolength\abovedisplayskip{#1}
  \addtolength\belowdisplayshortskip{#1}
  \addtolength\belowdisplayskip{#1}}
{\egroup\ignorespacesafterend}

\usepackage{subeqnarray}
\usepackage{cases}
\usepackage{mathtools} 
\usepackage{bm}
\usepackage{amsmath,amssymb,amsfonts}
\usepackage{textcomp}
\usepackage{xcolor}
\def\BibTeX{{\rm B\kern-.05em{\sc i\kern-.025em b}\kern-.08em
    T\kern-.1667em\lower.7ex\hbox{E}\kern-.125emX}}
\begin{document}
\newtheorem{Thm}{\textbf{Theorem}}
\newtheorem{Lem}{\textbf{Lemma}}
\newtheorem{Def}{\textbf{Definition}}
\newtheorem{Rem}{\textbf{Remark}}
\newtheorem{Exam}{\textbf{Example}}
\newtheorem{Sup}{\textbf{Assumption}}
\newtheorem{Cor}{\textbf{Collary}}

\title{State Variation Mining: On Information Divergence with Message Importance in Big Data\\
\thanks{
We indeed appreciate the support of the National Natural Science Foundation of China (NSFC) No. 61771283.}
}

\author{\IEEEauthorblockN{1\textsuperscript{st} Rui She}
\IEEEauthorblockA{\textit{Department of Electronic Engineering} \\
\textit{Tsinghua University}\\
Beijing, P.R. China \\
sher15@mails.tsinghua.edu.cn}
\and
\IEEEauthorblockN{2\textsuperscript{nd} Shanyun Liu}
\IEEEauthorblockA{\textit{Department of Electronic Engineering} \\
\textit{Tsinghua University}\\
Beijing, P.R. China \\
liushany16@mails.tsinghua.edu.cn}
\and
\IEEEauthorblockN{3\textsuperscript{rd} Pingyi Fan}
\IEEEauthorblockA{\textit{Department of Electronic Engineering} \\
\textit{Tsinghua University}\\
Beijing, P.R. China \\
fpy@tsinghua.edu.cn}
}

\maketitle

\begin{abstract}
Information transfer which reveals the state variation of variables usually plays a vital role in big data analytics and processing. 
In fact, the measures for information transfer could reflect the system change by use of the variable distributions,
similar to KL divergence and Renyi divergence.
Furthermore, in terms of the information transfer in big data, small probability events usually dominate the importance of the total message to some degree.
Therefore, it is significant to design an information transfer measure based on the message importance which emphasizes the small probability events.
In this paper, we propose a message importance transfer measure (MITM) and
investigate its characteristics and applications on three aspects.
First, the message importance transfer capacity based on MITM is presented to offer an upper bound for the information transfer process with disturbance.
Then,
we extend the MITM to the continuous case and
discuss the robustness by using it to measuring information distance.
Finally,
we utilize the MITM to guide the queue length selection
in the caching operation of mobile edge computing.
\end{abstract}


\section{Introduction}
Recently, the amount of data is exploding rapidly and the computing complexity for data processing is also increasing.
To some degree, this phenomenon is resulted from more and more mobile devices as well as the growing service of clouds.
In the literature, it is favored to process the collected data to dig out the hidden important information.
On one hand, it is necessary to improve the computation platforms for big data processing, such as cloud computing, fog computing and mobile edge computing (MEC).
On the other hand, a series of algorithmic technologies for big data analysis and mining are required, such as neural networks and machine learning, as well as distributed parallel computing, etc.

In many scenarios of big data, the small probability events attract more attention than the large probability ones. 
That is, the rarity of small probability events has higher value in use.
For instance, on anti-terrorist activities, there are only a few illegal people and hazardous agent that should be supervised especially \cite{Counterterrorism-systems}.
Moreover, in terms of the synthetic ID detection, it just focuses on a small number of artificial identities for financial frauds \cite{A-comprehensive-survey}.
Actually, how to mine and characterize small probability events becomes more challenging and more significant in modern life.

From the perspective of information theory, small probability events detection can be regarded as a kind of clustering problem. 
In particular, a graph-based rare category detection was presented based on the global similarity matrix \cite{Graph-based-rare}.
Furthermore, a time-flexible rare category detection was also designed by resorting to the time-evolving of graphs \cite{Rare-category-detection}.
In spite of these efficient algorithms for some special applications, it is worth noting that  they were designed by traditional information measures and theory, which originate from the viewpoint of typical events, namely the large probability events.

\subsection{Review of Message Importance Measure}
As two fundamental measures in information theory, Shannon entropy and Renyi entropy play a crucial role in many applications including communication engineering, estimation theory, hypothesis testing and pattern recognition.
However, they are not suitable enough for small probability events mining in the big data scenarios.
To do this, the message importance measure (MIM), a new information measure, is proposed to reflect the significance of small probability events.
Thus, let us review the definition of MIM briefly first \cite{message-importance-measure-and-its-application-to-minority-subset-detection-in-big-data}.\par
In a finite alphabet, for a given probability distribution $P=\{ p(x_1), p(x_2),..., p(x_n)\}$, the MIM with importance coefficient $\varpi \geq 0$ is defined as
\begin{equation}
L(P,\varpi) = \log\big\{ \sum\limits_{x_i} p(x_i)e^{\varpi\left(1-p(x_i)\right)} \big\},
\end{equation}
which measures the information importance of the distribution.
Then, by setting the parameter $\varpi=1$ and simplifying the form of MIM, it is easy to obtain its fundamental definition as follows.

\begin{defn}\label{defn:MIM}
For the discrete probability $P$=$\{p(x_1)$, $p(x_2)$, ...,$ p(x_n)\}$, the MIM can be given by
\begin{equation}\label{MIM_discrete1}
 \begin{aligned}
    L(P)
    & = \sum\limits_{x_i} p(x_i) e^{-p(x_i)}.\\
 \end{aligned}
\end{equation}
\end{defn}
Comparing with Shannon entropy and Renyi entropy, the MIM replaces the corresponding logarithm operator or polynomial operator with the exponential form so that the weight factors of small probability elements can be amplified much more.
This can help to reflect the significance of small probability events from the viewpoint of information measure.

In addition, as a kind of evaluation criteria, Fadeev's postulates are commonly used to describe the information measures including Shannon entropy and Renyi entropy \cite{On-measures-of-entropy-and-information}.
In this case, for two independent random distributions $P$ and $Q$, Renyi entropy has a weaker postulate than Shannon entropy, that is
\begin{equation}
    H(PQ) = H(P) + H(Q),
\end{equation}
where the function $H(\cdot)$ denotes the corresponding information measure.
Similarly, the MIM has a much weaker postulate than Renyi entropy, as follows
\begin{equation}
    H(PQ) \le H(P) + H(Q).
\end{equation}
Therefore, in the sense of generalized Fadeev's postulates, the MIM can be reasonably viewed as a kind of information measure similar to Shannon entropy and Renyi entropy.

\subsection{Message Importance Transfer Measure}
For an information transfer process, we consider such a model that all the $P$ and $Q$ satisfies the Lipschitz condition as follows,
\begin{equation}\label{eq.Lipschitz}
    |H(P)- H(Q)| \le \lambda\|P-Q\|_{1},
\end{equation}
where $P$ and $Q$ denote the original probability distribution and the final one respectively in the information transfer process; $\lambda>0$ is the Lipschitz constant; $H(\cdot)$ denotes a kind of information measure function; $\| \cdot \|_{1}$ denotes the $l_1$-norm measure.

Here, we shall investigate and measure information transfer process by use of the message importance.
Actually, how to characterize the message importance variation in the processing of big data is a critical and interesting problem.
On account of Definition \ref{defn:MIM},
it is available to regard the MIM as an element to measure the message importance variation for a dynamic system.
Then, a new information transfer measure based on the MIM is defined as follows.
\begin{defn}\label{defn:MID}
For two discrete probability $Q=\{q(x_1), q(x_2),$ $... , q(x_n)\}$ and $P=\{p(x_1), p(x_2), ... , p(x_n)\}$ satisfying the constraint in Eq. (\ref{eq.Lipschitz}), the message importance transfer measure (MITM) is defined as
\begin{equation}\label{MID_discrete}
 \begin{aligned}
    & D_{I}(Q||P)
    = \sum\limits_{x_i} \{ q(x_i) e^{-q(x_i)} -p(x_i) e^{-p(x_i)} \}.\\
 \end{aligned}
\end{equation}
\end{defn}
Note that the Definition \ref{defn:MID} characterizes the information transfer from the statistics.
That is, we can make use of MITM to measure the change of message importance focusing on small probability events in an information transfer process.

Actually, there exist a variety of different information measures handling the problem of information transfer process.
Shannon entropy and Renyi entropy are applicable to intrinsic dimension estimation \cite{On-local-intrinsic-dimension-estimation}. As well, the NMIM can be used in anomaly detection \cite{Non-parametric-Message-Important-Measure}.
Moreover, the directed information 
and Schreiber’s transfer entropy \cite{Measuring-information-transfer} are commonly applied to inferring the causality structure and characterizing the information transfer process.
In addition, referring to the idea from dynamical system theory, new information transfer measures are proposed to explore and exploit the causality between states in the system control \cite{Causality-preserving-information-transfer-measure}.

However, in spite of numerous kinds of information measures, few works focus on how to characterize the information transfer from the perspective of message importance in big data.
To this end, the MITM different from the above information measures is introduced.

\subsection{Organization}
We organize the rest of this paper as follows.
In Section II,
we introduce the message importance transfer capacity measured by the MITM to describe the information transfer with disturbance.
In Section III,
we extend the MITM to the continuous case to investigate the variation of message importance in the information transfer process.
In Section IV,
the MITM is used to discuss the queue length selection for the data caching in MEC from the viewpoint of queue theory.
Moreover, some simulations are presented to validate our theoretical results.
Finally, we conclude it in Section VI.
%

\section{Message Importance Transfer Capacity Based on Message Importance Transfer Measure}
In this section, we will introduce the MITM to characterize the information transfer process 
shown in Fig. \ref{fig_transfer_system}. To do so, we define the message importance transfer capacity measured by the MITM as follows.
\begin{figure}[!t]
\centering
\includegraphics[width=3.6in]{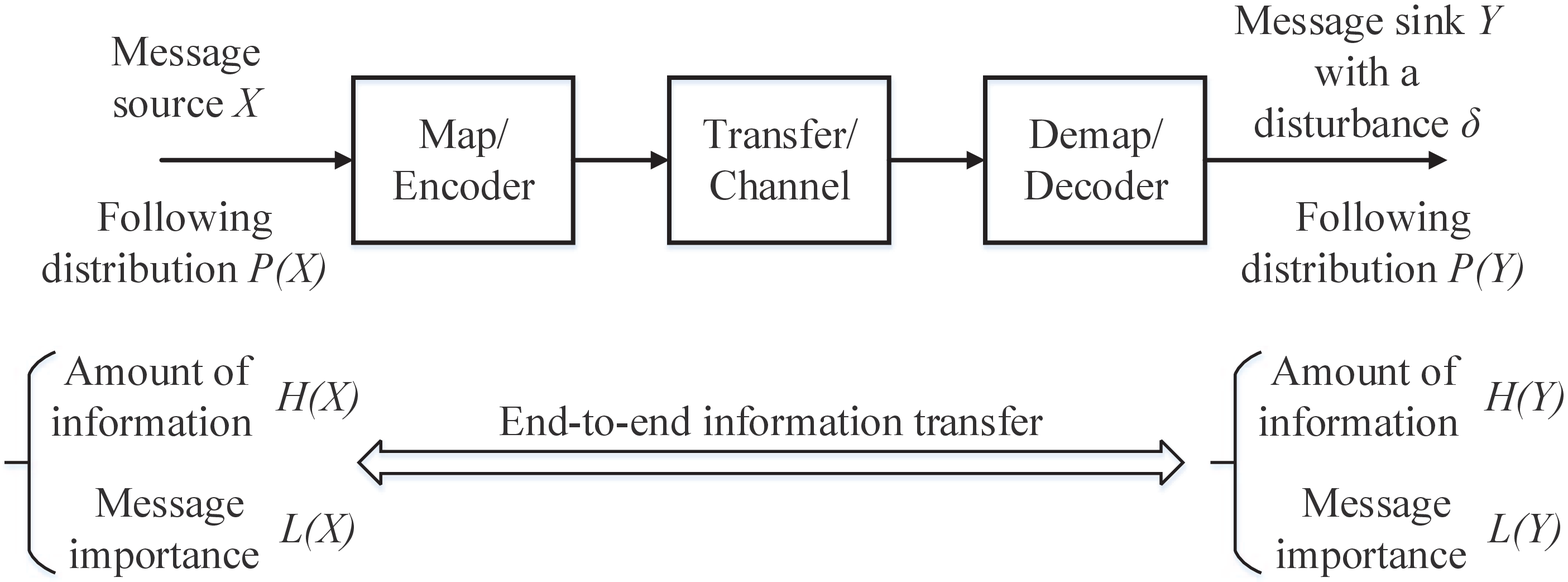}
\caption{Information transfer system model.}
\label{fig_transfer_system}
\end{figure}

\begin{defn}\label{defn:C_D}
Assume that there exists
an information transfer process (from the variable $X$ to $Y$) as,
\begin{equation}\label{eq.relation_1}
 \begin{aligned}
    \big\{ X, p(y-\delta_0|x), Y-\delta_0 \big| \delta_0 \in \{ \delta | \delta \sim p(\delta)\} \big\},
 \end{aligned}
\end{equation}
where $\delta$ denotes a disturbance following distribution $p(\delta)$ and $\delta_0$ is a certain element from the support set of $\delta$. In brief, Eq. (\ref{eq.relation_1}) can also be written as
\begin{equation}\label{eq.relation}
 \begin{aligned}
    \big\{ X, p(\tilde y|x), \tilde Y \big\},
 \end{aligned}
\end{equation}
where $\tilde y=y-\delta_0$ and $\tilde Y=Y-\delta_0$. Furthermore, $p(\tilde y|x)$ denotes a probability distribution matrix describing the information transfer from the variable $X$ following the distribution $p(x)$ to $\tilde Y$ following the distribution $p(\tilde y)$.
We define the message importance transfer capacity as
\begin{equation}
   \begin{aligned}\label{eq.D_channel_average}
    & C = \sum\limits_{\delta_0 \in \{\delta | \delta \sim p(\delta) \}} p(\delta_0) \tilde C(\delta_0), \\
   \end{aligned}
\end{equation}
in which
\begin{equation}
   \begin{aligned}\label{eq.D_channel}
    & \tilde C(\delta_0)= \max\limits_{p(x)} \{ L(\tilde Y) - L(\tilde Y|X)\}, \\
   \end{aligned}
\end{equation}
where $p(\tilde y_j) = \sum\limits_{x_i} p(x_i)p(\tilde y_j|x_i)$, $L(\tilde Y)=\sum\limits_{\tilde y_j} p(\tilde y_j)e^{-p(\tilde y_j)}$, $L(\tilde Y|X) = \sum_{\tilde y_j}\sum_{x_i} p(x_i,\tilde y_j)e^{-p(\tilde y_j|x_i)}$ with the constraint
$|L(\tilde Y)- L(\tilde Y|X)| \le \lambda\|p(\tilde y)-p(\tilde y|x)\|_{1}$.
\end{defn}

In order to have an insight into the applications of message importance transfer capacity, some specific information transfer scenarios are discussed as follows.

\subsection{Binary symmetric information transfer}
\begin{prop}\label{prop.symmetric}
Assume that there exists an information transfer process as same as that mentioned in Eq. (\ref{eq.relation_1}) and Eq. (\ref{eq.relation}),
where the disturbance $\delta$ follows a binary uniform distribution (namely $p$($\delta$)= (1/2, 1/2)), and the information transfer matrix is
\begin{equation}
 \begin{aligned}
   p(\tilde y|x) = \left [
   \begin{matrix}
    1-\beta & \beta \\
    \beta & 1-\beta
   \end{matrix}
   \right ],
 \end{aligned}
\end{equation}
which indicates that variables $X$ and $\tilde Y$ both obey the binary distributions.
In this case, the message importance transfer capacity is
\begin{equation}
\begin{aligned}
    C(\beta) = e^{ -\frac{1}{2} } - L(\beta),
\end{aligned}
\end{equation}
where $L(\beta)= \beta e^{-\beta} + (1-\beta)e^{-(1-\beta)} $ ($0<\beta<1 $) and
$|C(\beta)| \le \lambda\|p(\tilde y)-p(\tilde y|x)\|_{1}$
($\lambda \ge \frac{e^{ -\frac{1}{2} } - \beta e^{-\beta} + (1-\beta)e^{-(1-\beta)}}{|1-2\beta|} $).
\end{prop}

\begin{proof}
Considering a variable $X$ following the binary distribution $(p, 1-p)$, it is not difficult to see that
\begin{equation}
\begin{aligned}
     L(\tilde Y|X)
    & = \beta e^{-\beta} + (1-\beta)e^{-(1-\beta)}.
\end{aligned}
\end{equation}
Moreover, according to Eq. (\ref{eq.D_channel_average}) and Eq. (\ref{eq.D_channel}), we have
message importance transfer capacity as
\begin{equation}
\begin{aligned}
     & C(p, \beta) = \max\limits_{p} \Big\{ [p+ \beta(1-2p)]e^{-[p+ \beta(1-2p)]} \\
     & + [(1-p)+ \beta(2p-1)]e^{-[(1-p)+ \beta(2p-1)]} \Big\} - L(\beta).\\
\end{aligned}
\end{equation}
Then, it is readily seen that
\begin{equation}
\begin{aligned}
     & \frac{\partial C(p, \beta)}{\partial p} = (1-2\beta)\Big\{ [1-p-\beta(1-2p)]e^{-[p+ \beta(1-2p)]} \\
     & - [1-(1-p)-\varepsilon(2p-1)]e^{-[(1-p)+ \beta(2p-1)]} \Big\}. \\
\end{aligned}
\end{equation}

In the light of the monotonically decreasing of $\frac{\partial C(p, \beta)}{\partial p}$ for $p \in [0,1]$, it is apparent that $p=1/2$ is the only solution for $\frac{\partial C(p, \beta)}{\partial p} =0 $. Therefore, 
the proposition can be testified.
\end{proof}

According to Proposition \ref{prop.symmetric}, on one hand, when $\beta=1/2$, that is, the information transfer process is just random, we will gain the lower bound of $C(\beta)$, namely $C(\beta) =0$. On the other hand, when $\beta=0$,
we will have the maximum message importance transfer capacity.

\subsection{Strongly symmetric information transfer}
\begin{cor}
Assume that the information transfer process described by Eq. (\ref{eq.relation_1}) and Eq. (\ref{eq.relation}),
has a strongly symmetric information transfer matrix
\begin{equation}
 \begin{aligned}
   p(\tilde y|x) = \left [
   \begin{matrix}
    1-\beta & \frac{\beta}{K-1} &...& \frac{\beta}{K-1} \\
    \frac{\beta}{K-1} & 1-\beta & ...& \frac{\beta}{K-1} \\
    ...& ... & ... & ...\\
    \frac{\beta}{K-1} &...& \frac{\beta}{K-1} &  1-\beta
   \end{matrix}
   \right ],
 \end{aligned}
\end{equation}
and its disturbance $\delta$ follows an uniform distribution (namely $p$($\delta$)= (1/K,... 1/K)),
which indicates that variables $X$ and $\tilde Y$ both follow $K$-ary distributions.
Then, we have the message importance transfer capacity as
\begin{equation}
\begin{aligned}
    C(\beta) = e^{-\frac{1}{K}}- \{ (1-\beta)e^{- (1-\beta) } + \beta e^{-\frac{\beta}{K-1} } \},
\end{aligned}
\end{equation}
where the parameter $\beta \in (0,1)$
and $|C(\beta)| \le \lambda\|p(\tilde y)-p(\tilde y|x)\|_{1}$
($\lambda \ge \frac{ e^{-{1}/{K}}- (1-\beta)e^{- (1-\beta) } - \beta e^{-{\beta}/{K-1} }}{2|1-\beta-1/K|} $).
\end{cor}

\begin{proof}
This Corollary is an extension of Proposition \ref{prop.symmetric}. First, on account of the information transfer matrix and the Eq. (\ref{MIM_discrete1}), we have
\begin{equation}
\begin{aligned}
     L(\tilde Y|X)
    & = \beta e^{-\frac{\beta}{K-1}} + (1-\beta)e^{-(1-\beta)}.
\end{aligned}
\end{equation}

Then, similar to the proof of Proposition \ref{prop.symmetric}, we can also use Lagrange multiplier method to obtain the message information transfer capacity. In this case, the distribution of $\tilde Y$ should satisfy $p(\tilde y_1)=p(\tilde y_2)=...=p(\tilde y_K)=1/K$.


In addition, consider that the probability distribution of variable $X$ is $\{p(x_1),p(x_2),...,p(x_K) \}$. In the strongly symmetric transfer matrix, if the variable $X$ follows uniform distribution, namely $p(x_1)=p(x_2)=...=p(x_K)=1/K$, we will have
\begin{equation}
\begin{aligned}
    p(\tilde y_j)
    & = \sum\limits_{i=1}^{K} p(x_i, \tilde y_j)= \sum\limits_{i=1}^{K}p(x_i) p( \tilde y_j|x_i)\\
    & = \frac{1}{K} \sum\limits_{i=1}^{K} p(\tilde y_j|x_i)= \frac{1}{K},
\end{aligned}
\end{equation}
which indicates that $\tilde Y$ also follows the uniform distribution.

Therefore, it is testified that when the variable $X$ follows an uniform distribution which leads to the uniform distribution for variable $\tilde Y$, we will obtain the message importance transfer capacity $C(\beta)$.
%
%
%
\end{proof}

\section{Message Importance Transfer Measure in Continuous Cases}
Similar to the definition \ref{defn:MIM} and \ref{MIM_discrete1}, we can extend the two definition to the case with continuous distributions as follows
\begin{equation}\label{MIM_continue1}
 \begin{aligned}
    L(f(x))
    & = \int_{S_x} f(x) e^{-f(x)}dx, \quad \quad x \in S_x,\\
 \end{aligned}
\end{equation}
\begin{equation}\label{MID_continue}
 \begin{aligned}
    D_{I}(g(x)||f(x))
    & = L(g(x))-L(f(x)) \\
    & = \int_{S_x} { g(x) e^{-g(x)}- f(x) e^{-f(x)} }dx , x \in S_x,\\
 \end{aligned}
\end{equation}
where $g(x)$ and $f(x)$ are two probability distributions with respect to the variable $X$ in a given interval $S_x$. Moreover, $L(f(x))$ and $D_{I}(g(x)||f(x))$ can be
regarded as the continuous MIM and MITM.

Then, we investigate the variation of message importance by using
the continuous MITM,
which can also reflect the robustness of
continuous MITM.
Consider the observation model, $\mathcal{P}_{g_0|f_0}$: $f_0(x) \to g_0(x)$, that denotes an information transfer map for the variable $X$ from the probability distribution $f_0(x)$ to $g_0(x)$.
By using the similar way in \cite{An-information-theoretic-approach},
the relationship between $f_0(x)$ and $g_0(x)$ can be described as
 \begin{equation}\label{gx0_fx0}
    g_0(x)= f_0(x) + \epsilon f_0^{\alpha}(x)u(x),
 \end{equation}
and the constraint condition satisfies
 \begin{equation}\label{condition.gx0_fx0}
    \int_{S_x}\epsilon f_0^{\alpha}(x)u(x) dx=0,
 \end{equation}
where $\epsilon$ and $\alpha$ are adjustable coefficients. $u(x)$ is a perturbation function of the variable $X$ in the interval $S_x$.

Then, by using the above model, the end-to-end information distance measured by
continuous MITM
is given as follows.
\begin{prop}\label{prop.disturbance}
For two probability distributions $g_0(x)$ and $f_0(x)$ whose relationship satisfies the conditions Eq. (\ref{gx0_fx0}) and Eq. (\ref{condition.gx0_fx0}), the information distance measured by
continuous MITM
is given by
\begin{equation}
 \begin{aligned}
    & D_{I}(g_0(x)||f_0(x))\\
    & = \int_{S_x} \left\{ g_0(x) e^{-g_0(x)}- f_0(x) e^{-f_0(x)} \right\} dx \\
    & = \epsilon \sum\limits_{i=1}^{\infty}\frac{(-1)^i(i+1)}{i!}  \int_{S_x} f_0^{i+\alpha}(x)u(x) dx\\
    & + \frac{\epsilon^2}{2} \sum\limits_{i=1}^{\infty}\frac{(-1)^i(i+1)}{(i-1)!} \int_{S_x} f_0^{i-1+2\alpha}(x)u^2(x) dx + o(\epsilon^2),
 \end{aligned}
\end{equation}
where $\epsilon$ and $\alpha$ denote parameters, $u(x)$ is a function of the variable $X$ in the interval $S_x$, $|D_{I}(g_0(x)||f_0(x))| \le \int_{S_x}|\epsilon f_0^{\alpha}(x)u(x)| dx$ which satisfies the constraint Eq. (\ref{eq.Lipschitz}).
\end{prop}
In fact, Proposition \ref{prop.disturbance} describes the perturbation between $f_0(x)$ and $g_0(x)$.
Furthermore, we can obtain the
continuous MITM
between two distributions $g_1^{(u)}$ and $g_2^{(u)}$ based on the same reference distribution $f_0(x)$, which is given by
\begin{equation}
 \begin{aligned}
    & D_{I}(g_1^{(u)}(x)||g_2^{(u)}(x))\\
    & = [L(g_1^{(u)}(x))-L(f_0(x))]-[L(g_2^{(u)}(x))-L(f_0(x))]\\
    & = \epsilon \sum\limits_{i=1}^{\infty}\frac{(-1)^i(i+1)}{i!} \int_{S_x} f_0^{i+\alpha}(x)[u_1(x) - u_2(x)]dx\\
    & + \frac{\epsilon^2}{2} \sum\limits_{i=1}^{\infty}\frac{(-1)^i(i+1)}{(i-1)!} \int_{S_x} f_0^{i-1+2\alpha}(x) [u_1^2(x)-u_2^2(x)] dx\\
    & + o(\epsilon^2),
 \end{aligned}
\end{equation}
where the $\epsilon$ and $\alpha$ are parameters, $u_1(x)$ and $u_2(x)$ are functions of the variable $X$,
and
\begin{equation}
    g_1^{(u)}(x)= f_0(x) + \epsilon f_0^{\alpha}(x)u_1(x), \quad \forall x\in S_x,
\end{equation}
\begin{equation}
    g_2^{(u)}(x)= f_0(x) + \epsilon f_0^{\alpha}(x)u_2(x), \quad \forall x\in S_x,
\end{equation}
with the constraint $|D_{I}(g_1(x)||g_2(x))| \le \int_{S_x}|\epsilon f_0^{\alpha}(x)\{u_{1}(x)-u_{2}(x)\}| dx$.

It is apparent that
when the parameter $\epsilon$ is small enough, the
continuous MITM
is convergent with the order of $O(\epsilon)$. Actually, this provides a way to apply
the continuous MITM
to measure the variantion of message importance,
if the system does not have relatively large change.

\section{ Application in Mobile Edge Computing with the M/M/s/k queue}
Consider the MEC system that consists of numerous mobile users, an edge server, and a central cloud. 
The queue model on the edge server can be considered as the M/M/s/k queue, where the first and the second $M$ denote the request interarrival time of mobile users and service request time in the edge server respectively, and both of them follow exponential distribution; $s$ is the parallel processing core number; $k$ denotes the
queuing buffer size
\cite{multi-objective-optimization-for-computation-offloading}.

In order to save resources of system, we now consider a more complicated M/M/s/k model which has the request lose depending on the queue length, namely the real arrival rate satisfies $\tilde \lambda_j=\tilde \lambda \cdot h_j$ ($\tilde \lambda$ is the original arrival rate and the parameter $h_j= \frac{1}{1+j}$ depends on the queue length $j$) \cite{An-explicit-solution,On-multiserver-feedback}.
In fact, the state probability of this queue model is derived from the stationary process, namely a dynamic equilibrium based on birth and death process.
In this case, we can obtain the steady queue state probability $p_{k,j}$ ($j=0,..., s+k$) as follows
\begin{equation}\label{eq.p_0}
    p_{k,0}= \Big[ \sum\limits_{j=0}^{s-1} \frac{a^j}{j!j!} + \frac{a^s}{s!} \cdot \sum\limits_{j=s}^{s+k} \frac{\rho^{j-s} }{j!} \Big]^{-1}, \ \
\end{equation}
\begin{equation}\label{eq.p_k_j}
    p_{k,j} = \frac{a^j}{j!j!}p_{k,0}, \quad  (0<j<s),\qquad \quad \ \
\end{equation}
\begin{equation}\label{eq.p_k_s}
    p_{k,j} = \frac{a^s}{s!j!} \rho^{j-s} p_{k,0}, \quad ( s \le j \le s+k),
\end{equation}
where
$s$ is the number of servers, $k$ is the buffer or caching size, the traffic intensity $\rho= a /s$ as well as $a= \tilde \lambda / \tilde \mu$ ($\tilde \lambda$ and $\tilde \mu$ are the original arrival rate and service rate respectively).

As for the MITM,
it can be used to distinguish the state probability distributions in the above M/M/s model.
By use of Taylor series expansion, the approximate MIM is given by
\begin{equation}\label{eq.MIM_MMsk}
\begin{aligned}
    & \sum\limits_{j=0}^{s+k} p_{k,j} e^{-p_{k,j}}
     = \sum\limits_{j=0}^{s+k} p_{k,j}[1-p_{k,j} + O(p_{k,j}^2)]\\
    & \doteq 1 -  p_{k,0}^2 \bigg\{ \sum\limits_{j=0}^{s-1} {(\frac{a^j}{j!j!})^2}+ (\frac{a^s}{s!})^2 \sum\limits_{j=s}^{s+k} (\frac{\rho^{j-s}}{j!})^2  \bigg\}. \\
\end{aligned}
\end{equation}
Then, referring to Eq. (\ref{eq.MIM_MMsk}), we can use MITM to characterize the message importance gap for the M/M/s model as follows.
\begin{prop}\label{pro.MID_queue}
As for the M/M/s model mentioned in Eq. (\ref{eq.p_0})-(\ref{eq.p_k_s}), the information difference between two queue state probability distributions $P_{k}= \{{p}_{k,0}, {p}_{k,1}, ..., {p}_{k,s+k}, 0, 0,..., 0\}$ and $P_{k+1} = \{{p}_{k+1,0},$ ${p}_{k+1,1}, ..., {p}_{k+1,s+k+1}, 0,..., 0\}$ with buffer size $k$ and $k+1$ respectively, can be measured by MITM as
\begin{equation}\label{MID_queueing}
 \begin{aligned}
    & D_I(P_{k+1}||P_{k})\\
    & = \sum\limits_{j=0}^{s+k+1} {p}_{k+1,j} e^{-{p}_{k+1,j}} - \sum\limits_{j=0}^{s+k}{p}_{k,j} e^{-{p}_{k,j}} \\
    & \doteq \Big\{ \frac{1}{(\varphi_1 + \varphi_2 \sum\limits_{j=s}^{s+k}\frac{\rho^{j-s}}{j!})^2}
    - \frac{1}{(\varphi_1 + \varphi_2 \sum\limits_{j=s}^{s+k+1}\frac{\rho^{j-s}}{j!})^2} \Big\} \\
    & \quad \cdot \Big\{ \sum\limits_{j=0}^{s-1} {(\frac{a^j}{j!j!})^2}
    + \varphi_2^2 \sum\limits_{j=s}^{s+k} (\frac{\rho^{j-s}}{j!})^2 \Big\} \\
    & \quad - \frac{ \varphi_2^2 \rho^{2k+2} }{ [(s+k+1)!]^2
    \big(\varphi_1^2+ \varphi_2^2\sum\limits_{j=s}^{s+k} \frac{\rho^{j-s}}{j!}\big) },
 \end{aligned}
\end{equation}
where $p_{k,j}$ and $ p_{k+1,j}$ are queue state probability in the M/M/s/{k} and M/M/s/{k+1} models with the constraint $|D_I(P_{k+1}||P_{k})| \le \lambda\|P_{k+1}-P_{k}\|_{1}$,
as well as the parameter $\varphi_1$ and $\varphi_2$ are given by $ \varphi_1= \sum_{j=0}^{s-1} {a^j}/{(j!j!)}$ and $\varphi_2= {a^s}/{s!}$.
\end{prop}

Similarly, it is not difficult to derive the MITM between the queue state probability distributions $P_{\infty} = \{ {p}_{\infty,0}, {p}_{\infty,1}, ...,$ ${p}_{\infty,\infty} \}$ and $P_{k}= \{{p}_{k,0}, {p}_{k,1}, ..., {p}_{k,s+k}, 0, 0,..., 0\}$ with buffer size $\infty$ and $k$, which is given by
\begin{equation}\label{MID_queueing_infty}
 \begin{aligned}
    & D_I(P_{\infty}|| P_{k}) \\
    & \doteq \Big\{ \frac{1}{(\varphi_1 + \varphi_2 \sum\limits_{j=s}^{s+k}\frac{\rho^{j-s}}{j!})^2}
    - \frac{1}{\big[ \varphi_1 + \varphi_2 (\frac{e^{\rho}}{\rho^s}
    - \sum\limits_{j=0}^{s-1}\frac{\rho^{j-s}}{j!}) \big]^2} \Big\}\\
    & \quad \cdot \Big\{ \sum\limits_{j=0}^{s-1} {(\frac{a^j}{j!j!})^2}
    + \varphi_2^2 \sum\limits_{j=s}^{s+k} (\frac{\rho^{j-s}}{j!})^2 \Big\}\\
    & \quad - \frac{ \varphi_2^2 \Big(\frac{e^{\rho}}{\rho^s}
    - \sum\limits_{j=0}^{s+k}\frac{\rho^{j-s}}{j!}\Big)  }{ \Big[ \varphi_1 + \varphi_2 (\frac{e^{\rho}}{\rho^s}
    - \sum\limits_{j=0}^{s-1}\frac{\rho^{j-s}}{j!}) \Big]^2}.
 \end{aligned}
\end{equation}

Moreover, for the queue length selection, it is required that the distinction between two distribution $P_{\infty}$ and $P_{k}$ should be small enough, namely, $|D_{I}(P_{\infty}|| P_{k})| \le \epsilon $ ($\epsilon$ is a small parameter). Since that the lower bound of buffer size is complicated, we have a looser lower bound as follows
\begin{equation}
\begin{aligned}
    & k \ge
     \frac{ \ln \Big\{ 1-\frac{1-\rho}{\varphi_2} \big[ ({\varphi}/{\sum\limits_{j=0}^{s-1} (\frac{a^j}{j!j!})^2})^{-\frac{1}{2}} -\varphi_1 \big] \Big\} }{\ln \rho} -1,
\end{aligned}
\end{equation}
where the parameter $\varphi$ is given by
\begin{equation}
\begin{aligned}
   \varphi= \epsilon +{\sum\limits_{j=0}^{s-1}(\frac{a^j}{j!j!})^2+\frac{\varphi_2^2 e^{\rho}}{\rho^s} }{\Big[ \varphi_1 + \varphi_2 (\frac{e^{\rho}}{\rho^s}
        - \sum\limits_{j=0}^{s-1}\frac{\rho^{j-s}}{j!}) \Big]^{-2} }.
\end{aligned}
\end{equation}
It is easy to see that $\epsilon$ plays a key role in the caching size selection when using finite size caching to imitate the infinite caching working mode.

Similar to MITM, the KL divergence between the queue state probability distributions with buffer size $k+1$ and $k$ is given by
\begin{equation}
 \begin{aligned}
    & D(P_{k} || P_{k+1})\\
    & = \sum\limits_{j} p_{k,j} \log \frac{1}{p_{k+1,j}} - \sum\limits_{j} p_{k,j} \log \frac{1}{p_{k,j}}\\
    & = \log \Big\{ 1+ \frac{  \rho^{k+1}}{(s+k+1)!\big(\varphi_1+\varphi_2 \sum\limits_{j=s}^{s+k} \frac{\rho^{j-s}}{j!}\big)} \Big\},
 \end{aligned}
\end{equation}
where the parameters $p_{k,j}$, $p_{k+1,j}$, $\varphi_1$ and $\varphi_2$ are the same as them in Proposition \ref{pro.MID_queue}.

Likewise, we can derive the KL divergence between the queue state distributions with buffer size $\infty$ and $k$ as
\begin{equation}
 \begin{aligned}
    D(P_{k}||P_{\infty})
    & = \log  \frac{ \varphi_1 + \varphi_2 (\frac{e^{\rho}}{\rho^s}- \sum_{j=0}^{s-1} \frac{\rho^{j-s}}{j!})}{ \varphi_1+ \varphi_2 \sum_{j=s}^{s+k}\frac{\rho^{j-s}}{j!}  }.
 \end{aligned}
\end{equation}

For the queue length selection with KL divergence, we have a looser lower bound of buffer size as follows
\begin{equation}
 \begin{aligned}
    k \ge \frac{\ln \Big\{ 1- \frac{(1-\rho)}{ 2^{\epsilon} \varphi_2}\big[ \varphi_1(1-2^{\epsilon})
    + \varphi_2(\frac{e^{\rho}}{\rho^s} - \sum\limits_{j=0}^{s-1} \frac{\rho^{j-s}}{j!})
     \big] \Big\}}{\ln \rho} -1.
 \end{aligned}
\end{equation}

To validate our derived results in theory, some simulations are presented.
The events arrivals are listed
in Table \ref{table_arrival}.
It is readily seen that they have the same average interarrival time as $1/\tilde \lambda_{j,0}$.
Besides, the traffic intensity is selected as $\rho = 0.9$ in all discussed cases.

\begin{table}[!t]
\renewcommand{\arraystretch}{1.3}
\caption{The Interarrival Time Distributions of Events' Arrivals}
\label{table_arrival}
\newcommand{\tabincell}[2]{\begin{tabular}{@{}#1@{}}#2\end{tabular}}
\centering
\begin{tabular}{|c|c|c|c|}
\hline
\tabincell{l} {
Type of \\
Distribution} &
\tabincell{l} {
Exponential \\
Distribution } &
\tabincell{l} {
Uniform \\
Distribution} &
\tabincell{l} {
Normal \\
Distribution }
\\
\hline
$P(X)$ &
$X \sim E({\tilde \lambda_{j,0}})$ &
$X \sim U(0,2/{\tilde \lambda_{j,0}})$ &
$X \sim N(\frac{1}{\tilde \lambda_{j,0}}, \frac{1}{{\tilde \lambda_{j,0}}^2} )$
\\
\hline
\end{tabular}
\end{table}


\begin{figure}[!tp]
\centering
\subfigure[The simulation and theoretical results of information measures]{\includegraphics[width=3.0in]{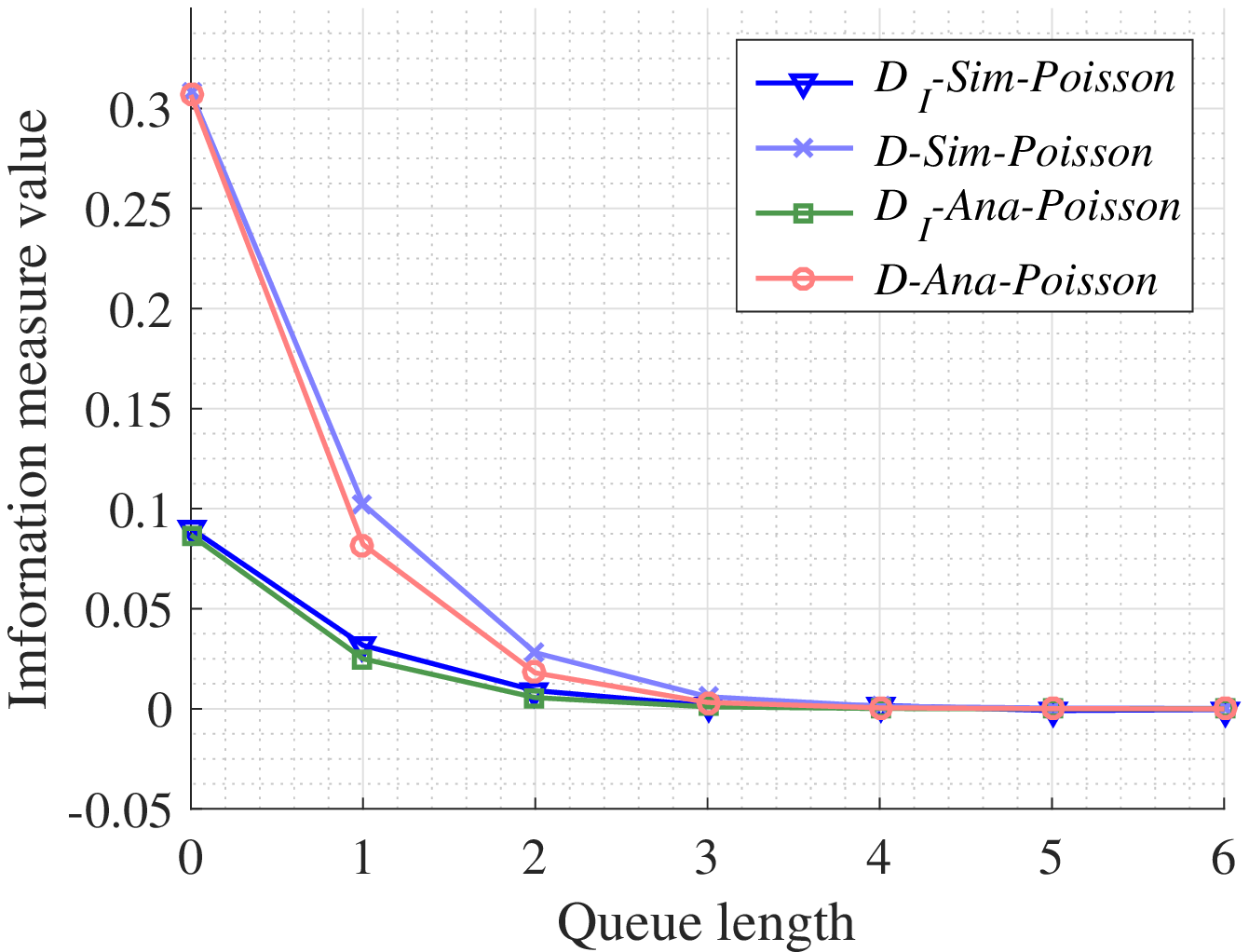}%
\label{fig_kk_arrival_ana}}
\hfil
\subfigure[Information measures for different arrival events distributions]{\includegraphics[width=3.0in]{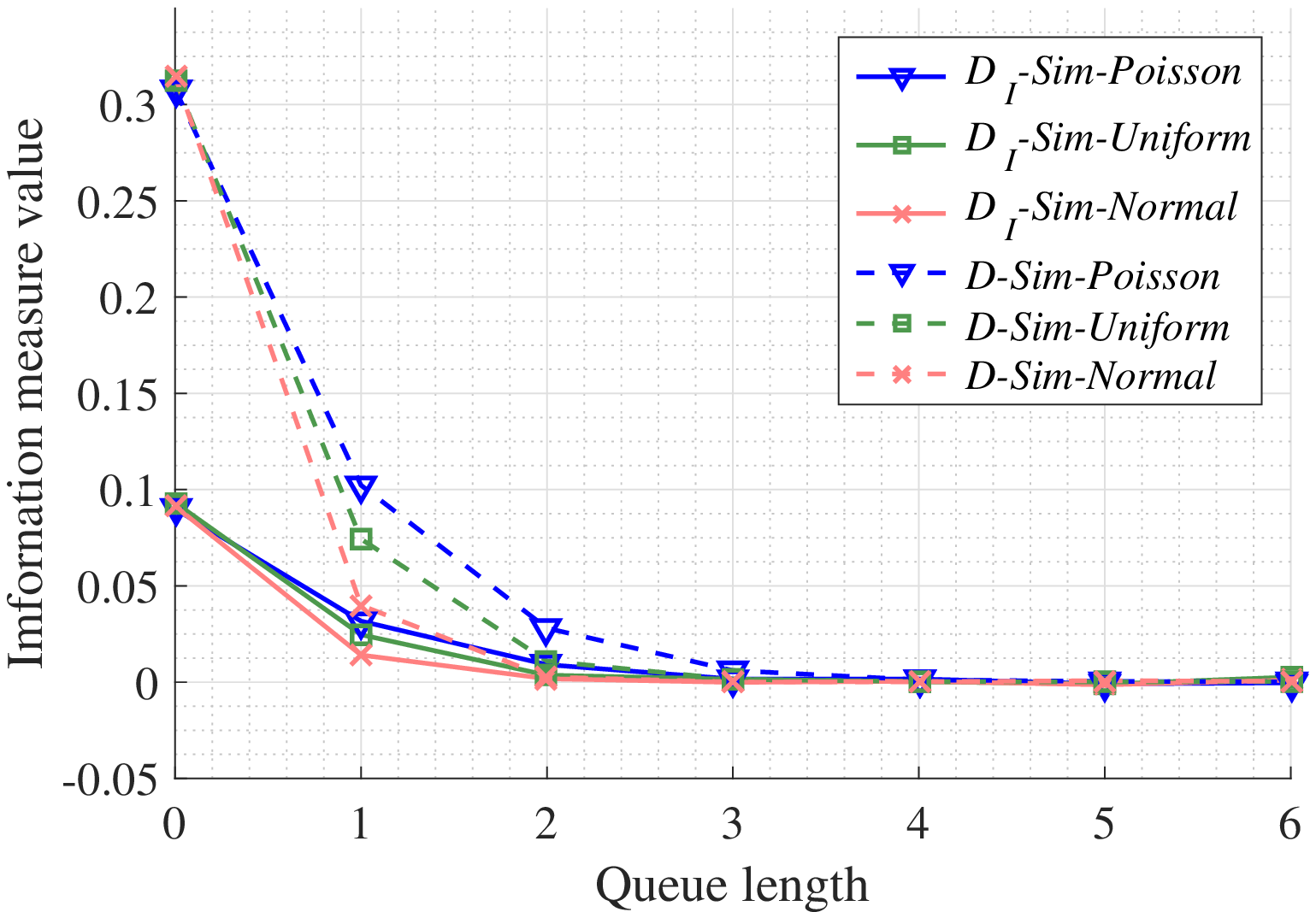}%
\label{fig_kk_arrival}}
\caption{The performance of information measures for the state variation between the queue length $k$ and $k+1$ in the case of server number $s=1$. }
\label{fig_kk_performance}
\end{figure}

In Fig. \ref{fig_kk_performance} and \ref{fig_inftyk_performance}, the legends $D_I$-$Sim$, $D_I$-$Ana$ and $D$-$Sim$, $D$-$Ana$ denote the simulation results and the analytical results for MITM and KL divergence, respectively. It is illustrated that the convergence of MITM is faster than that of KL divergence, which indicates that MITM may provide a reasonable lower bound to select the caching size for MEC.
In addition, we can see that the Poisson distribution corresponds the worst case for the arrival process among the three discussed cases with respect to the convergence of both MITM and KL divergence.

\section{Conclusion}
In this paper, we investigated the information transfer problem in big data and proposed an information measure, i.e., MITM. 
Furthermore, this information measure has its own dramatic characteristics on paying more attention to the message importance hidden in big data.
This makes the information measure as a promising tool for information transfer measure in big data.
We presented the message importance transfer capacity measured by the MITM which can give an upper bound for the information transfer with disturbance.
Furthermore, the MITM was extended to the continuous case to investigate the variation of message importance in the information transfer process.
In addition,
we employed the MITM to discuss the caching size selection in the MEC.


\begin{figure}[!tp]
\centering
\subfigure[The simulation and theoretical results of information measures]{\includegraphics[width=3.0in]{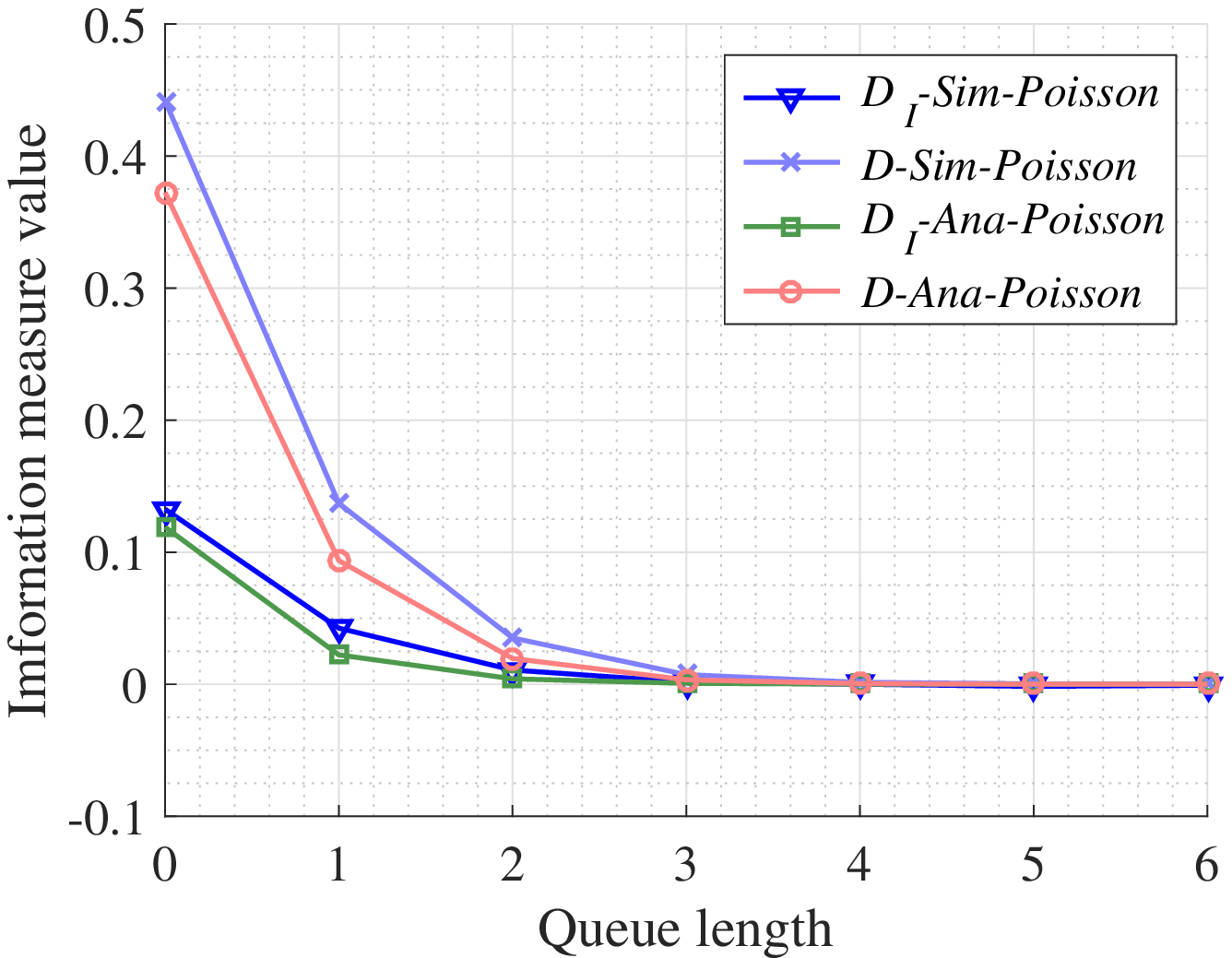}%
\label{fig_inftyk_arrival_ana}}
\hfil
\subfigure[Information measures for different arrival events distributions]{\includegraphics[width=3.0in]{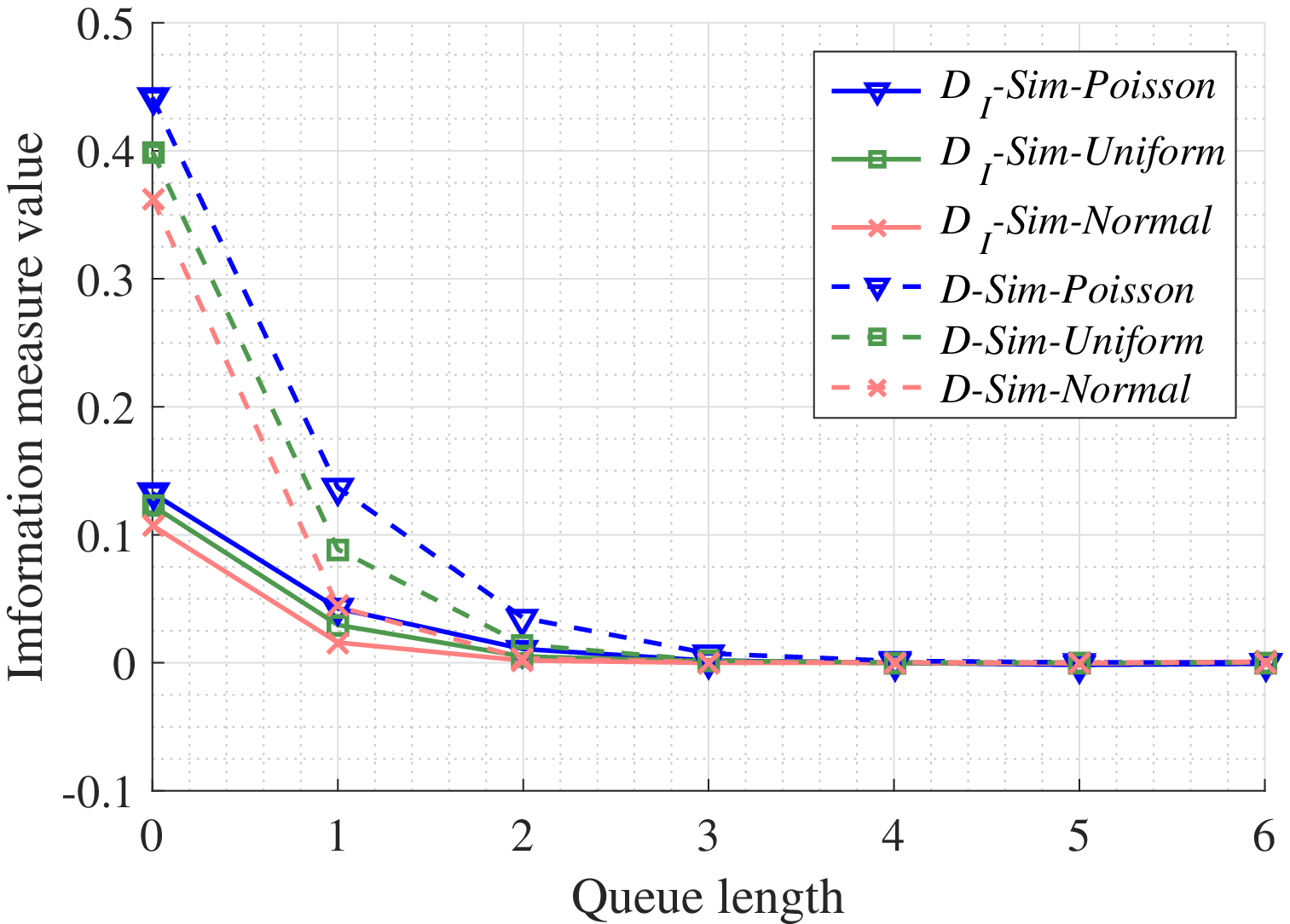}%
\label{fig_inftyk_arrival}}
\caption{The performance of information measures for the state variation between the queue length $k$ and $\infty$ in the case of server number $s=1$. }
\label{fig_inftyk_performance}
\end{figure}

%

\end{document}